\documentclass[a4paper]{article}

\usepackage{mathrsfs}
\usepackage{amsthm}
\usepackage[dvips]{graphicx}

\newtheorem{proposition}{Proposition}[section]

\begin{document}

\title{\textsc{on automorphism groups of networks}}
\author{\textsc{Ben D. MacArthur} \\ Bone and Joint Research Group, University of Southampton, \\ Southampton General Hospital, Southampton, SO16 6YD, UK \\ \\
\textsc{Rub\'en J. S\'anchez-Garc\'\i{}a}\footnote{Ben MacArthur and Rub\'en S\'anchez-Garc\'\i{}a contributed equally to this work.} \\ Mathematisches Institut, Heinrich-Heine Universit\"at \\ D\"usseldorf, Universit\"atsstr 1, 40225, D\"usseldorf, Germany \\ \\ \textsc{James W. Anderson} \\ School of Mathematics, University of Southampton,  \\ Southampton, SO17 1BJ, UK} 

\maketitle

\begin{abstract}
We consider the size and structure of the automorphism groups of a variety of empirical `real-world' networks and find that, in contrast to classical random graph models, many real-world networks are richly symmetric. We relate automorphism group structure to network topology and discuss generic forms of symmetry and their origin in real-world networks. 
\end{abstract}

\begin{center}
\textsc{keywords: Complex Network, Automorphism Group}
\end{center}

\newpage
\section{\label{sec:intro} Introduction}
The use of complex networks to model the underlying topology of `real-world' complex systems -- from social interaction networks such as scientific collaboration networks\cite{newman2,newman3} to biological regulatory networks\cite{albert3} and technological networks such as the internet\cite{siganos} -- has attracted much current research interest\cite{albertSM,newman,strogatz}. Previous studies have highlighted the fact that seemingly disparate networks often have certain features in common including (amongst others): the `small-world' property\cite{wattsSW}; the power-law distribution of vertex degrees\cite{barabasiPA}; and network construction from motifs\cite{milo}. 

Identification of universal structural properties such as these allows generic network properties to be decoupled from system-specific features. In this present work we consider the symmetry structure of a variety of real-world networks and find that a certain degree of symmetry is also ubiquitous in complex systems.

We consider network symmetry via the automorphism group of the underlying graph. Firstly, we identify `essential' network symmetries and use these symmetries to derive a natural direct product decomposition of the automorphism group into irreducible factors. This decomposition is \textit{per se} a very efficient way to handle large automorphism groups of real-world networks. We associate with each factor in this decomposition a \emph{symmetric subgraph} -- the subgraph on which the factor subgroup acts non-trivially -- and investigate the generic structure of symmetric subgraphs.

\section{\label{sec:auto}Network Automorphism Groups}
Mathematically, a \emph{network} is a graph, $\mathscr{G}=\mathscr{G}(V,E)$, with vertex set, $V$ (of size $N_\mathscr{G}$), and edge set, $E$ (of size $M_\mathscr{G}$) where vertices are said to be \emph{adjacent} if there is an edge between them. An \emph{automorphism} is a permutation of the vertices of the network which preserves adjacency. The set of automorphisms under composition forms a group, $\textrm{Aut}(\mathscr{G})$, of size $a_\mathscr{G}$\cite{bollobasMGT}. Throughout this discussion we shall let $\mathscr{G}$ refer to a generic network, and $G$ to a generic group. If the network is a multi-digraph, we remove weights and directions and consider the automorphism group of the \textit{underlying} graph.

Here, the \texttt{nauty} program\cite{mckay} -- which includes one of the most efficient graph isomorphism algorithms available\cite{foggia} -- is used to calculate the size and structure of the various automorphism groups.

Table \ref{auttable} gives the order of the automorphism group of some real-world complex networks, all of which are nontrivial. Since almost all large graphs (including, for example, the classical Erd{\"o}s-R{\'e}nyi random graphs) are asymmetric\cite{bollobasRG} this symmetry is somewhat unexpected and begs an explanation.

Many networks -- for example the internet and the world wide web -- are `growing'\cite{barabasiPA} (that is, new vertices are added to the network over time). Generically, any growth process in which allows for new vertices to be added to the network \emph{one at a time} naturally leads to a network with locally tree-like regions. Such locally tree-like areas are common in real-world networks and their presence is important because, while the majority of large graphs are asymmetric, it is common for large random \emph{trees} to exhibit a high degree of symmetry\cite{harary}, deriving from the presence of identical branches about the same fork. Thus we expect a certain degree of tree-like symmetry to be present in many real-world networks. In the following sections we determine the extent to which real-world symmetry is locally tree-like. We begin by considering the structure of network automorphism groups.

\begin{table}[t]
\begin{center}
\begin{tabular}[c]{l  c  c  c}
Network & $N_\mathscr{G}$ & $M_\mathscr{G}$ & $a_\mathscr{G}$ \\
\hline 
\hline
Human B Cell Genetic Interactions\cite{basso} & $5,930$ & $64,645$ & $5.9374 \times 10^{13}$ \\
\textit{C. elegans} Genetic Interactions\cite{zhong} & $2,060$ & $18,000$ & $6.9985 \times 10^{161}$ \\
BioGRID datasets\cite{biogrid}: \\
\hspace{2.15cm} Human & $7,013$ & $20,587$ & $1.2607 \times 10^{485}$ \\
\hspace{1.4cm} \textit{S. cerevisiae} & $5,295$ & $50,723$ & $6.8622 \times 10^{64}$ \\
\hspace{1.6cm} \textit{Drosophila} & $7,371$ & $25,043$ & $3.0687 \times 10^{493}$ \\
\hspace{1.15cm} \textit{Mus musculus} & $209$ & $393$ & $5.3481 \times 10^{125}$ \\
Internet (Autonomous Systems Level)\cite{caida} & $22,332$ & $45,392$ & $1.2822 \times 10^{11,298}$ \\
US Power Grid\cite{wattsSW} & $4,941$ & $6,594$ & $5.1851 \times 10^{152}$ \\
\\
\end{tabular}
\end{center}
\caption{\label{auttable}\small{The size of the automorphism group of some real-world networks. The size of the automorphism group of the largest connected component is given (to $5$ significant figures) which, in all cases, contains at least $93\%$ of the vertices in the network. Connected components were extracted using \texttt{Pajek}\cite{pajek}.}} 
\end{table}

\section{\label{sec:struct} Factorization of Automorphism Groups}
In this section we describe a computationally efficient factorization of large network automorphism groups -- which are often too large to allow direct analysis (see Table \ref{auttable}) -- into `irreducible building blocks'.

Consider the permutations of a set of $n$ points $X =\{ x_1, \ldots, x_n
\}$. The \emph{support} of a permutation $p$ is the set of points which $p$ moves, $\textup{supp}(p) = \{ x_i \, | \, p(x_i) \neq x_i \}$. Two permutations $p$ and $q$ are \emph{disjoint} if their supports are non-intersecting. If $p$ and $q$ are disjoint then they commute (with
respect to the composition of permutations). Similarly, two sets of permutations $P$ and $Q$ are
\emph{support-disjoint} if every pair of permutations $p \in P$ and
$q \in Q$ have disjoint supports.

Let $\mathscr{G}$ be a network with automorphism group $\textrm{Aut}(\mathscr{G})$. Let $S$ be a set of generators of $\textrm{Aut}(\mathscr{G})$. Suppose that we partition $S$ into $n$ support-disjoint subsets $S = S_1 \cup
\ldots \cup S_n$ such that each $S_i$ cannot itself be decomposed into smaller support-disjoint subsets. Call $H_i$ the subgroup generated by $S_i$. Since $S$ is a generating set and elements from different factors $H_i$, $H_j$ commute, this procedure gives a direct product decomposition:
\begin{equation} \label{decompbody}
\textrm{Aut}(\mathscr{G}) \cong H_1 \times H_2 \times \ldots \times H_n.
\end{equation}
Note that, in general, the choice of generators of a group is not unique and different choices of generating sets may give different decompositions. Thus, for this decomposition to be well-defined, we need to show that it is unique and the factors in Eq. (\ref{decompbody}) are `irreducible'; that is, they cannot be written as $K \times L$ with $K$ and $L$ support-disjoint subgroups.

A group $G$ is \emph{support-indecomposable} if it cannot be written as $K \times L$ with $K$ and $L$ support-disjoint subgroups. Similarly, a set $S$ is support-indecomposable if it cannot be written as $S_1 \cup S_2$ with $S_1, S_2 \neq \emptyset$ both support-disjoint subsets. 

\begin{proposition} \label{prop1}
The subgroups in Eq.(\ref{decompbody}) are independent of the choice of generators (that is, unique) and support-indecomposable (that is `irreducible') when the generating set $S$ satisfies the following two conditions
\begin{itemize}
  \item[$(*)$] $S$ does not contain elements in the form $s=gh$ with
  $g,h \neq 1$ and $g,h$ support-disjoint;
  \item [$(**)$] if a subset $S' \subset S$ generates a subgroup $H
  \le G$ such that $H = H_1 \times H_2$ with $H_1$ and $H_2$ support-disjoint then
  there exits a partition $S'=S_1\cup S_2$ such that $S_i$
  generates $H_i$.
\end{itemize}
\end{proposition}

Note that these conditions are ensured if, for example, the \texttt{nauty} algorithm is used to calculate the generators of $\textrm{Aut}(\mathscr{G})$ (see parts (1) and (2) of Theorem 2.34 in \cite{mckay}). The proof of proposition \ref{prop1} can be considered in two parts: irreducibility and uniqueness.  

\begin{proposition}(irreducibility)
Let $S$ be a finite set of permutations and $H$ the group generated by $S$. If $H$ is support-indecomposable as a group, then so is $S$ as a set. The converse is also true when $S$ satisfies $(*)$.
\end{proposition}

\begin{proof} The first claim is clear. For the converse, suppose that $S=\{s_1,\ldots,s_n\}$ is support-indecomposable as a set but $H = K \times L$ ($K,L$ support-disjoint). Then $s_1 = kl$ for $k\in K$, $l\in L$. By condition $(*)$, $k=1$ or $l=1$, that is, $s_1 \in K$ or $s_1 \in L$, and similarly for $s_2,\ldots,s_n$. Thus $S = (S\cap K) \cup (S\cap L)$. Since $S$ is support-indecomposable as a set, $S\cap K$ or $S\cap L$ is empty, that is, $S \subseteq K$ or $S \subseteq L$. Hence $H=K$ and $L=1$ or $H=L$ and $K=1$.
\end{proof}

\begin{proposition}(uniqueness)
Suppose that $X$ and $Y$ are two sets of generators of a permutation group $G$, with associated direct product decompositions \[
    \begin{array}{rcl}
        G &\cong& H_1 \times \ldots \times H_n,\\
        G &\cong& K_1 \times \ldots \times K_m\,.
    \end{array}
\]
If both $X$ and $Y$ are essential, then $n=m$ and there is a permutation $\sigma$ of the factors such that $H_i \cong K_{\sigma(i)}$ for $i=1,\ldots, n$.
\end{proposition}

\begin{proof}(\emph{sketch}) 
Firstly, generalize condition $(**)$ to a finite number of subgroups $H_1,\ldots,H_n$, by induction on $n$. Then apply this to the first set of generators $X$ with respect to the second decomposition. We then have a partition $X = X_1 \cup \ldots \cup X_m$ such that $X_i$ generates $K_i$ ($1 \le i \le m$). Suppose that $H_1$ is generated by a set $\{x_1,\ldots,x_t\} \subseteq X$. Since $H_1$ and $X_1, \ldots, X_m$ are support-indecomposable, we must have $\{x_1,\ldots,x_t\} \subseteq X_{i_1}$ for some $i_1$. That is, $H_1 \subset K_{i_1}$. Since $K_{i_1}$ is support-indecomposable this implies $H_1 = K_{i_1}$. The same argument applies for $H_2,\ldots,H_n$.
\end{proof}

Thus, the decomposition given in Eq. (\ref{decompbody}) is well-defined if (for example) the \texttt{nauty} algorithm is used. We shall refer to this decomposition as the \emph{geometric} decomposition, and note that it is a simple variation of the Krull-Schmidt factorization into the direct product of indecomposable subgroups\cite{rotman}. A GAP\cite{gap} procedure which calculates the geometric decomposition for an arbitrary permutation group is available from the authors on request.

In general the geometric decomposition is coarser than the Krull-Schmidt decomposition since non-disjoint permutations may still commute. The Krull-Schmidt decomposition may easily be obtained from the geometric decomposition using a computational group theory package such as GAP. The main advantage to using the geometric decomposition is that it provides a computationally efficient way to calculate the structure of large real-world networks and relates more intuitively to graph topology than the Krull-Schmidt factorization.  For all the real-world networks we considered the automorphism group was factorized efficiently using this method, while a direct `brute-force' factorization was not computationally feasible. 

The geometric decompositions of some real-world networks are given in Table \ref{auttable1}. In all cases the geometric factors are either symmetric groups or wreath products of symmetric groups (wreath products are a mild generalization of direct products, see \cite{rotman} for a definition and examples). 

\textsc{Remark}: it is a result of P{\'o}lya that automorphism groups of trees belong to the class of permutation groups which contains the symmetric groups and is closed under taking direct and wreath products\cite{biggs}. Thus, the automorphism groups of many real-world networks belong to the same class of groups as the automorphism groups of trees. Note however, this does not necessarily mean that real-world symmetry is tree-like (for example, the complete graphs also belong to this class). In the following section we relate automorphism group structure to network topology in order to determine the extent to which real-world symmetry is, in fact, tree-like.   

\begin{table}[t]
\begin{center}
\begin{tabular}[c]{l  p{7cm} c} 
Network & $\textrm{Aut}(\mathscr{G})$ & \% Basic Factors \\
\hline
\hline
Human B Cell Genetic Interactions & $C_2^{36} \times S_3^2 \times S_4$ & 97.4\\
\textit{C. elegans} Genetic Interactions & $C_2^{95} \times S_3^{27} \times S_4^9 \times S_5^5 \times S_6^5 \times S_7^2 \times S_8 \times S_9^2 \times S_{10} \times S_{11} \times S_{33} \times (C_2 \wr C_2)$ & 98.7 \\
BioGRID datasets: \\
\hspace{2.15cm} Human & $C_2^{286} \times S_3^{80} \times S_4^{30} \times S_5^{14} \times S_6^{10} \times S_7 \times S_8^2 \times S_9^5 \times S_{10} \times S_{11} \times S_{12}^3 \times S_{15} \times S_{16}^2 \times S_{17} \times S_{23} \times S_{26} \times S_{44}$ & 99.5 \\ 
\hspace{1.4cm} \textit{S. cerevisiae} & $C_2^{42} \times S_3^{8} \times S_4^{5} \times S_5^{2} \times S_6^{2} \times S_7 \times S_{14} \times S_{17}$ & 100 \\
\hspace{1.6cm} \textit{Drosophila} & $C_2^{289} \times S_3^{86} \times S_4^{35}  \times S_5^{19} \times S_6^{11} \times S_7^{10} \times S_8^5 \times S_9^3 \times S_{10}^3 \times S_{11}^3 \times S_{12}^3 \times S_{14} \times S_{16} \times S_{20} \times S_{30}$ & 99.2 \\  
\hspace{1.15cm} \textit{Mus musculus} & $C_2^{7} \times S_3^{4} \times S_4  \times S_5 \times S_6^3 \times S_8 \times S_{10} \times S_{11} \times S_{12} \times S_{26} \times S_{44}$ & 100 \\
Internet (Autonomous Systems Level) & $C_2^{955} \times S_3^{352} \times S_4^{197}  \times S_5^{120} \times S_6^{83} \times S_7^{56}  \times S_8^{55} \times S_9^{47} \times S_{10}^{32} \times S_{11}^{24} \times S_{12}^{14} \times S_{13}^{13} \times S_{14}^{13} \times S_{15}^{8} \times S_{16}^{9} \times S_{17}^{7}  \times S_{18}^3 \times S_{19}^{12} \times S_{20}^{7} \times S_{21}^{10} \times S_{22} \times S_{23} \times S_{24}^6 \times S_{25} \times S_{26}^{3} \times S_{27}^{3}  \times S_{28} \times S_{29}^2 \times S_{30}^{2} \times S_{31}^{2} \times S_{32}^{2} \times S_{33} \times S_{34}^{4} \times S_{35} \times S_{36} \times S_{37}  \times S_{38}^2 \times S_{41} \times S_{42} \times S_{43}  \times S_{46} \times S_{48}^2 \times S_{50} \times S_{51} \times S_{52}^2 \times S_{54} \times S_{56} \times S_{58} \times S_{59} \times S_{60} \times S_{62} \times S_{64} \times S_{70} \times S_{71} \times S_{76} \times S_{79} \times S_{82} \times S_{95} \times S_{112} \times S_{137} \times S_{138} \times S_{147} \times S_{167} \times S_{170} \times S_{194} \times S_{202} \times S_{216} \times S_{276} \times S_{318} \times S_{356} \times (C_2 \wr C_2)^2$ &  98.4 \\
US Power Grid & $C_2^{228} \times S_3^{44} \times S_4^{14} \times S_5^4 \times S_6^2 \times S_7 \times S_9 \times (C_2 \wr C_2)^8 $ & 88.1 \\
\\
\end{tabular}
\end{center}
\caption{\label{auttable1} \small{The geometric decomposition of the automorphism group of some real-world networks, and percentage basic factors (see section \ref{sec:topology}). In all cases, the automorphism group can be decomposed into direct and wreath products of symmetric groups.}}
\end{table}

\section{\label{sec:topology}Automorphism Group Structure and Symmetric Subgraphs}
The \emph{induced subgraph} on a set of vertices $S \in \mathscr{G}$ is the graph obtained by taking $S$ and any edges whose end points are both in $S$. We define a \emph{symmetric subgraph} as the induced subgraph on the support of a geometric factor $H$ (that is, on the points with non-trivial action by $H$). It is natural to ask whether there are any properties of symmetric subgraphs which are generic. 

From table \ref{auttable1} it is clear that most of the geometric factors found in real-world networks are isomorphic to $S_n$ (for some $n$). Furthermore, almost all of these symmetric factors act transitively on their supports. We shall refer to transitive symmetric factors as \emph{basic factors} and associated symmetric subgraphs as \emph{basic symmetric subgraphs} (BSS's). We shall refer to all other factors as \emph{complex factors} and their associated symmetric subgraphs as \emph{complex symmetric subgraphs}. Table \ref{auttable1} shows that, in the cases we considered, almost all factors are basic and therefore that almost all symmetry is due to the presence of basic symmetric subgraphs. 

Since a graph $\mathscr{G}$ on $n$ vertices with $\textrm{Aut}(\mathscr{G}) \cong S_n$ is either empty or complete\cite{lauri} it is immediate that BSS's are also either empty or complete. Furthermore, transitivity ensures that for a given BSS $\mathscr{B}$ and a given vertex $v \in \mathscr{G-B}$, \emph{all} vertices in $\mathscr{B}$ are adjacent to $v$ or \emph{none} are. This means that most real-world symmetry is due to the presence of symmetric \emph{cliques} (complete subgraphs invariant under $\textrm{Aut}(\mathscr{G})$) and symmetric \emph{bicliques} (complete bipartite subgraphs invariant under $\textrm{Aut}(\mathscr{G})$). 

In practice, for all the real-world networks we considered, bicliques other than \emph{stars} (a $k$-star is a subgraph consisting of a vertex of degree $k+1$ adjacent to $k$ vertices of degree $1$), although occasionally present, were rare (see Fig. \ref{pnn} for some examples). In fact, we found that stars were the predominant symmetry structure present in all the networks we considered, although symmetric cliques were also significantly present in a number of networks. For example, the \emph{c. elegans} genetic regulatory network\cite{zhong} -- which was constructed by inferring connections from multiple datasets across multiple organisms and is thus arguably one of the most well-characterized biological networks available -- contains multiple symmetric cliques, including one on $33$ vertices corresponding to the largest geometric subgroup in the decomposition of its automorphism group. This example (and those in Fig. \ref{pnn}) illustrate the fact that although much real-world symmetry is tree-like (and thus can be related to generic growth processes) a certain degree is not. In particular, a significant proportion of real-world symmetry originates in symmetric cliques. Since cliques and bicliques are topologically very similar (they are both complete multipartite graphs), the presence of symmetric cliques in complex networks may derive from similar growth processes to those that produce stars in combination with local clustering. 

Fig. \ref{example} gives a typical arrangement of symmetric subgraphs (basic and complex) found in many real world networks, illustrating the relationship between these symmetric subgraphs and the structure of the network automorphism group. Since complex symmetric subgraphs can potentially take \emph{any} form it is not possible to say anything general about their structure. However, since they are rare they may be considered on a case-by-case basis. Fig. \ref{pnn} shows the complex symmetric subgraphs present in the US power grid, illustrating that in some real-world networks a certain degree of complex symmetry is present.

\begin{figure}[t]
\begin{center}
\includegraphics[width=0.7\textwidth]{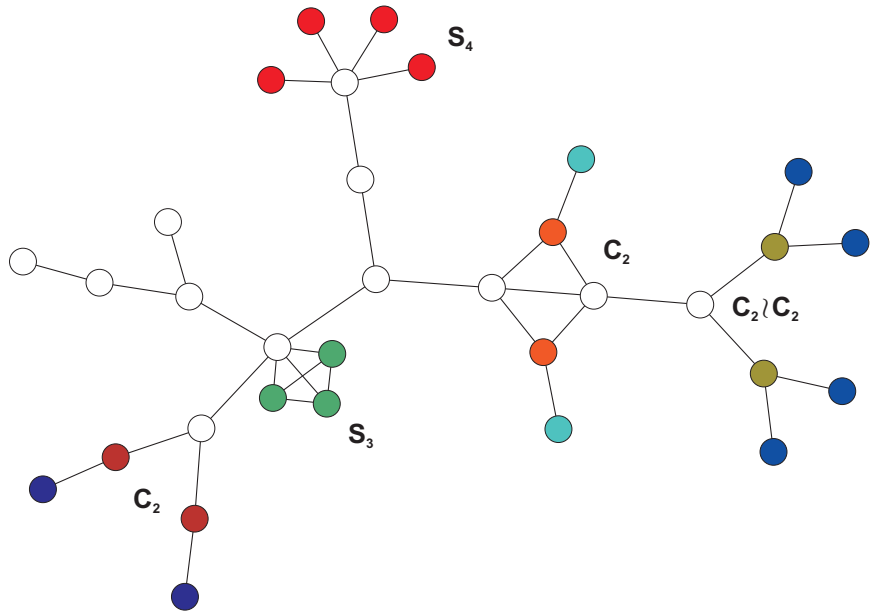}
\end{center}
\caption{\label{example}\small{\textbf{A typical arrangement of symmetric subgraphs}. The geometric decomposition of the automorphism group of this graph is $\textrm{Aut}(\mathscr{G}) \cong C_2^2 \times S_3 \times S_4 \times (C_2 \wr C_2)$. This example illustrates how different symmetric subgraphs contribute to the automorphism group, as well as showing common `non-treelike' real-world symmetry. In particular note the 4-star (red) and the 3-clique (green) which correspond to the factors $S_4$ and $S_3$ respectively in the geometric decomposition of $\textrm{Aut}(\mathscr{G})$. We found that wreath product factors generally associate with extended branches (see the far right of this figure), although this is not always the case (see the starred subgraph in Fig. \ref{pnn} for example). Vertices are colored by orbit, fixed points are in white.}}
\end{figure}

\begin{figure}[t]
\begin{center}
\includegraphics[width=0.75\textwidth]{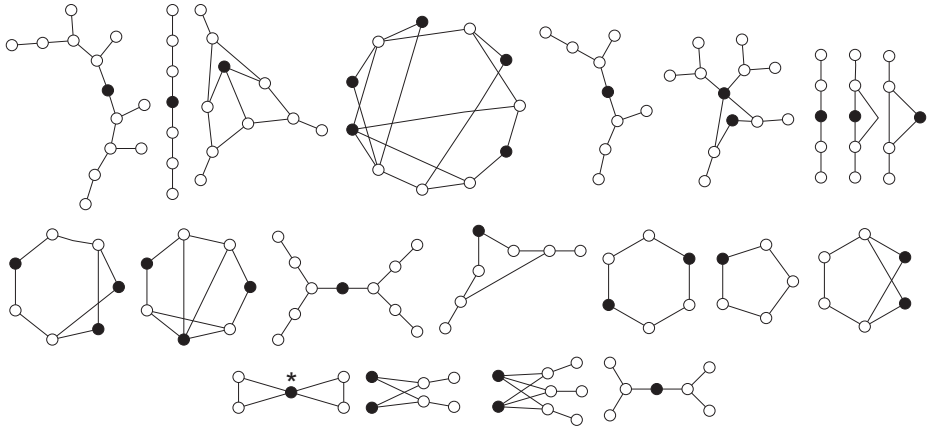}
\caption{\label{pnn}\small{\textbf{Complex symmetric subgraphs in the US power grid.} Vertices in white correspond to those in the symmetric subgraphs. Vertices in black are those adjacent to those in the symmetric subgraph, and are shown to clarify subgraph structure. The starred subgraph has automorphism group $C_2 \wr C_2$, illustrating that wreath products do not associate exclusively with extended branches such as in the example in Fig. \ref{example}}}
\end{center}
\end{figure} 

\section{\label{sec:conclusions} Conclusions}
We have considered the automorphism groups of a variety of real-world networks and have found that a certain degree of symmetry is ubiquitous. We have constructed a practical decomposition of the automorphism groups of these networks (the geometric decomposition), and found that the automorphism groups can typically be decomposed into direct and wreath products of symmetric groups. We have shown that each geometric factor can be associated with a symmetric subgraph, and demonstrated that most factors can be related to either a symmetric clique or symmetric biclique. Thus, we find that these two types of subgraph generically account for almost all real-world network symmetry. 

We anticipate that an important manifestation of real-world symmetry may be its effect on network \emph{behavior}. For example, many biological regulatory networks are not simply static but rather can be associated with an underlying dynamical system:
\begin{equation} \label{dynsys}
\frac{d x_i}{dt} = f_i(A_{ij}x_j),
\end{equation}
where $x_i$ is the state of the $i$th species (for example gene, protein etc.) and $A_{ij}$ is the network adjacency matrix. In many cases, while network topology may be known (from reverse engineering of microarray data, for instance) the specific form of the coupling functions $f_i$ are unknown or experimentally unverified (or unverifiable). In these circumstances it is useful to know if there are any model \emph{independent} properties of the dynamical system (that is properties which derive from network structure and are exhibited by a variety of coupling functions). It is here that symmetry is useful since, for equivariant dynamical systems (those which remain `unchanged' under the action of a given symmetry group), dynamic behavior can often be related in a generic way to symmetry structure. More precisely \emph{`the symmetries of a system imply a `catalogue' of typical forms of behavior from which the actual behavior is `selected''}\cite{golubitsky}. Since the symmetry group of Eqs. (\ref{dynsys}) is necessarily a subgroup of the underlying network automorphism group, knowledge of the network automorphism group is a first step toward classifying typical forms of behavior that a given network may exhibit. This has been achieved for certain simple systems including those with $S_n$ symmetry\cite{elmhirst,golubitsky} but has yet to be extensively investigated for more complex systems such as the ones we consider here. In light of these advances, and the fact that automorphism groups of real-world networks commonly have a rather simple generic form -- they can be decomposed into direct and wreath products of symmetric groups -- we anticipate that determination of a catalogue of forms of behavior for some complex networks may be feasible.

\section{Acknowledgments}
This work was funded by the EPSRC and by a London Mathematical Society scheme 6 grant. 

\bibliographystyle{amsplain}
\bibliography{autpapers_new}

\end{document}